\documentclass[12pt]{article}
\usepackage[fleqn]{amsmath}
\usepackage{amsthm,amsfonts,latexsym,amssymb,amscd}
\usepackage[mathscr]{eucal}
\def\Vir{{\mathrm {Vir}}}
\def\supp{{\mathrm {supp}}}
\def\id{{\mathrm {id}}}
\title{Construction of holomorphic local conformal framed nets}
\author{
{\sc Yasuyuki Kawahigashi}\footnote{Supported in part by Global 
OE Program ``The research and training center for new development
in mathematics'', the Mitsubishi Foundation Research Grants and
the Grants-in-Aid for Scientific Research, JSPS.}\\
{\small Department of Mathematical Sciences}\\
{\small The University of Tokyo, Komaba, Tokyo, 153-8914, Japan}
\\[0,05cm]
{\small and}
\\[0,05cm]
{\small Kavli Institute for the Physics and Mathematics of the Universe}\\
{\small 5-1-5 Kashiwanoha, Kashiwa, 277-8583, Japan}\\
{\small E-mail: {\tt yasuyuki@ms.u-tokyo.ac.jp}}
\\[0,3cm]
{\sc Noppakhun Suthichitranont}\footnote{Supported by MEXT scholaship}
\\
{\small Department of Mathematical Sciences}\\
{\small The University of Tokyo, Komaba, Tokyo, 153-8914, Japan}\\
{\small E-mail: {\tt noppakhuns@hotmail.com}}}
\newtheorem{theorem}{Theorem}[section]						

\newtheorem{lemma}[theorem]{Lemma}
\newtheorem{proposition}[theorem]{Proposition}

\newtheorem{definition}[theorem]{Definition}

\newtheorem{remark}[theorem]{Remark}

\begin{document}

\maketitle

\begin{abstract}
We construct holomorphic local conformal framed nets extended
from a tensor power of the Virasoro net with $c=1/2$ with a pair
of binary codes $(C,D)$ satisfying the conditions given by Lam
and Yamauchi for holomorphic framed vertex operator algebras.
Our result is an operator algebraic counterpart of theirs,
but our proof is entirely different. We apply the $\alpha$-induction
in order to identify the representation theory of
``code local conformal net'' and this gives rise to
the existence of the desired local conformal net.
\end{abstract}

\section*{Introduction}

The concept of states and observables are fundamental in quantum physics.
Here, the observables are the meaurements of physical quantities and
states are the physical quantities of interest. In algebraic quantum
field theory, states are vectors in a Hilbert space and observables
are self-adjoint operators defined on it. As the observation is done
in a certain spacetime domain, the observables are defined on each
specific spacetime domain. These ideas are translated into mathematical
structures using the language of operator algebras \cite{Araki, Haag}.
When we have conformal symmetry on the $1+1$ dimensional Minkowski 
spacetime, we have full conformal field theory. In the language of
operator algebra, we have a family of von Neumann algebras
$\{\mathcal{A}(\mathcal{O})\}_{\mathcal{O}}$ where the regions
$\mathcal{O}$ are double cones which are rectangles with the
sides parallel to the lines $x = \pm t$. The family
$\{\mathcal{A}(\mathcal{O})\}_{\mathcal{O}}$ satisfies isotony,
locality, conformal covariance, existence of vacuum vector and
positivity of energy where the locality is defined by spacelike
separation. We call the family
$\{\mathcal{A}(\mathcal{O})\}_{\mathcal{O}}$ a 2-dimensional
local conformal net.
We can restrict $\{\mathcal{A}(\mathcal{O})\}_{\mathcal{O}}$ to
the lines $x = \pm t$ and replace the spacelike separation by
non-intersection which is simpler.
(See \cite{Kawahigashi-Longo-Comm} for more details.)
Through this process, we have two {\sl chiral} conformal field theories.

Local conformal nets in chiral conformal field theory arise from
the restriction of the rectangles $\mathcal{O}$ to the line
 $x = \pm t$. They are of the form $\{\mathcal{A}(I)\}_{I}$ where
each $\mathcal{A}(I)$ is a von Neumann algebra abd $I$ is an open
interval in the ``spacetime'' which is now the circle $S^1$. We
have another mathematical approach to the same physical theory
which is the language of vertex operator algebra. Vertex operators
arise from the Fourier series expansion of operator-valued
distributions on $S^1$. Conceptually, it should be possible to
choose certain vertex operators and test functions with the supports
in the interval $I$, and form a family of mathematical objects
equivalent to local conformal nets. Unlike local conformal nets
where the state spaces are Hilbert spaces, there is no assumption
on the existence of an inner product and the completeness on 
the state spaces in the theory of vertex operator algebra. This
motivates an open problem whether these two mathematical objects
correspond to each other or not, i.e., if an example of one of the
objects is given, the counterpart in the form of the other object
should be found. There are several pieces of evidence which
emphasize that the answer to this open problem is positive,
at least under some extra nice conditions. Both mathematical
objects have examples associated to affine Lie algebras
and the Virasoro algebra.
Given a positive definite lattice $L$, one can construct a
lattice vertex operator algebra and their twisted orbifolds as
in \cite{Frenkel-Lepowsky-Meurman}. The construction of local
conformal nets for a positive definite lattice $L$ is given
in \cite{Staszkiewicz}. An operator algebraic counterpart of the
twisted orbifolds of lattice vertex operator algebras are also
constructed by extending lattice conformal nets in \cite{Dong-Xu}.

For a given central charge $c < 1$, Virasoro conformal nets
$\Vir_{c}$ and simple Virasoro vertex operator algebras $L(c,0)$ are
essentially the same objects in many aspects. They have the same
representation theory and their irreducible representations obey the
same fusion rules. Our main interest is in the simple current extensions
of their tensor products which are holomorphic (defined to have
only one irreducible representation which is the trivial one).
By definition,
framed vertex operator algebras contain $L(1/2 , 0)^{\otimes n}$ as a
subalgebra with the same conformal element
\cite{Dong-Griess-Hoehn, Lam-Yamauchi, Miyamoto}. In \cite{Lam-Yamauchi},
Lam and Yamauchi extended the tensor products of the simple Virasoro
vertex operator algebra with central charge $1/2$ to holomorphic framed
vertex operator algebras using a pair of even binary codes $(C , D)$
where $C$ is the dual code of $D$ and $D$ satisfies the following
conditions:
\begin{enumerate}
\item The length of $D$ is $16n$ where $n$ is a positive integer;
\item The Hamming weights of all elements in $D$ are divisible by $8$;
\item The all-one word is in $D$.
\end{enumerate}
This kind of extension exists for $L(1/2 , 0)^{\otimes 16n}$. Here,
the word $(c_1 , c_2 , \dots, c_{16n})$ in $C$ correspond to the
module $\bigotimes^{16n}_{i=1} L(1/2 , c_i /2 )$ and the word
$(d_1 , d_2 ,\dots, d_{16n})$ in $D$ correspond to the module
of $\bigoplus_{c \in C}
\bigotimes^{16n}_{i=1} L(1/2 , c_i /2 )$ 
whose decompositions into $L(1/2 , 0)^{\otimes 16n}$-modules have
the conformal weights $1/16$ for the entries $d_i = 1$.

In this manuscript, we give an operator algebraic counterpart of
a holomorphic framed vertex operator algebra which is a holomorphic
local conformal framed net in the sense of
\cite{Kawahigashi-Longo-Advances}.
A local conformal framed net $\mathcal{A}$ is
defined as an irreducible extension of $\Vir_{1/2}^{\otimes n}$.
The structure of local conformal framed nets as the simple current
extension $\Vir_{1/2}^{\otimes n}
\rtimes \mathbb{Z}^{k}_2\rtimes \mathbb{Z}^{l}_2$ for some
integers $k$ and $l$ is shown in
\cite[Theorem 4.3]{Kawahigashi-Longo-Advances}, and the
corresponding result for vertex operator algebras has been
obtained in \cite{Lam-Yamauchi}. We construct holomorphic
local conformal framed nets which are extended from tensor products of
Virasoro conformal nets with central charge $1/2$ by using a pair of
even binary codes $(C , D)$ satisfying the above conditions
for $D$ with $C = D^{\perp}$. To show the existence of
$\Vir_{1/2}^{\otimes 16n} \rtimes C \rtimes D$, we need to know the
representation theory of $\Vir_{1/2}^{\otimes 16n} \rtimes C$,
which is our counterpart of the code vertex operator algebra
and decide whether an appropriate action
of $D$ on $\Vir_{1/2}^{\otimes n} \rtimes C$ exists or not.
(See \cite{Miyamoto} and references therein for code vertex
operator algebras.)

A representation of a local conformal net corresponds to a unitary
equivalence class of endomorphisms called a sector. For chiral
conformal field theory, a sector $\rho$ consists of endomorphisms
$\rho_I$ for each $\mathcal{A}_I$ such that each $\rho_I$ leaves
$\mathcal{A}_{I'}$ invariant where $I'$ is the set of the interior
points of the complement of $I$. The local conformal net
$\Vir_{1/2}^{\otimes 16n}$ is completely rational in the sense
of \cite[Definition 8]{Kawahigashi-Longo-Mueger}. It has the
$\mu$-index which is the square sum of the statistical dimensions
of all irreducible sectors equal to $2^{32n}$. The $\mu$-index
of $\Vir_{1/2}^{\otimes 16n} \rtimes C$ is $2^{32n}/
\lvert C \rvert^{2}$ \cite{Kawahigashi-Longo-Mueger}. This
provides some conditions on the  statistical dimensions of all
irreducible sectors of $\Vir_{1/2}^{\otimes 16n} \rtimes C$.

The simple current extension $\Vir_{1/2}^{\otimes 16n} \rtimes C$
is an irreducible extension of $\Vir_{1/2}^{\otimes 16n}$. 
As defined by Longo and Rehren in \cite{Longo-Rehren}, the
$\alpha^{\pm}$-induction $\rho \mapsto \alpha^{\pm}_{\rho}$ transforms
a sector of $\Vir_{1/2}^{\otimes 16n}$ to a sector of
$\Vir_{1/2}^{\otimes 16n} \rtimes C$.
(See also \cite{Boeckenhauer-Evans, Xu} for general propertoes
of the $\alpha$-induction.) 
Note that the
$\alpha^{\pm}$-induction depends on the code $C$ as it is defined
using a canonical endomorphism $\gamma$ from
$\Vir_{1/2}^{\otimes 16n} \rtimes C$ into $\Vir_{1/2}^{\otimes 16n}$
for each spacetime domain.  Let $\theta$ be the restriction on
$\Vir_{1/2}^{\otimes 16n}$ of the canonical endomorphism $\gamma$.
The endomorphism $\theta$ decomposes into a direct sum of the
irreducible sectors of $\Vir_{1/2}^{\otimes 16n}$ corresponding
to $C$ with multiplicity $1$. Since the system of irreducible sectors
of $\Vir_{1/2}^{\otimes 16n}$ is non-degenerate in the sense
of \cite{Rehren}, we have that the irreducible sectors of
$\Vir_{1/2}^{\otimes 16n} \rtimes C$ are the intersection between
the subsectors of the $\alpha^{\pm}$-induction of the irreducible
sectors of $\Vir_{1/2}^{\otimes 16n}$. We can first consider
the $\alpha^{\pm}$-induction of the irreducible sectors $\rho$
such that $\alpha^{+}_{\rho} = \alpha^{-}_{\rho}$ which holds
if and only if
$\varepsilon^{\pm}(\rho, \theta)\varepsilon^{\pm}
(\rho, \theta) = \mathbf{1}$ where $\varepsilon^{\pm}$ is
the statistics operator which is unitary \cite{Boeckenhauer-Evans}.
Using \cite[part III: Lemma 3.8]{Boeckenhauer-Evans}, we study
the $S$-matrix of $\Vir_{1/2}^{\otimes 16n}$ to investigate
whether $\varepsilon^{\pm}(\rho, \sigma)\varepsilon^{\pm}
(\rho, \sigma) = \mathbf{1}$ for each irreducible subsector
$\sigma$ of $\theta$. Some of the irreducible sectors may give
the $\alpha^{\pm}$-induction identical to that of some other
irreducible sectors. Some of them split into a direct sum of
irreducible sectors of $\Vir_{1/2}^{\otimes 16n} \rtimes C$ as
the $\alpha^{\pm}$-induction does not preserve the irreducibility.
At this point, we apply the $\mu$-index of
$\Vir_{1/2}^{\otimes 16n} \rtimes C$ to identify the family of
irreducible sectors of $\Vir_{1/2}^{\otimes 16n} \rtimes C$. The
strategy was previously used in the identification of the
representation theory of $\Vir_{1/2}^{\otimes 2}
\rtimes \{ (0,0), (1,1)\}$ as given in
\cite[Proposition 2.2]{Kawahigashi-Longo-Advances}.

For the case of $\Vir_{1/2}^{\otimes 16n} \rtimes C \rtimes D$
with $C = D^{\perp}$, the problem is divided into three cases
according to the the cardinality of the triply even binary code $D$.
When the cardinality of $D$ is either $2$ or $4$, the problem is
similar to the case of $\Vir_{1/2}^{\otimes 2} 
\rtimes \{ (0,0), (1,1)\}$ and we can obtain the answer by applying
the same strategy directly. For $D$ with higher cardinality, we
have difficulties to deal with the contribution to the $\mu$-index
of the $\alpha^{\pm}$-induction since it may only give irreducible
sectors with the statistical dimensions more than $1$. When this
happens, the irreducible subsectors of such $\alpha^{\pm}$-induction
are not automorphisms. Thus, we prove by mathematical induction by
constructing a decreasing sequence of triply even codes
$D = D_1 \supset D_2 \supset \cdots \supset D_{p-1} \supset D_p =
\{(0)_{16n}, (1)_{16n} \}$ which gives an increasing
series of even codes $C = C_1 \subset C_2 \subset 
\cdots C_{p-1} \subset C_p$. Here, we choose $D_i$'s such
that $\lvert D_i \rvert / \lvert D_{i+1} \rvert$ is $2$ and $D_{p-1}$
is generated by $\beta$ and $(1,1,\dots,1)$ for some
$\beta$ different from the identity and the all-one word in $D$. We
draw a conclusion that a sector associated to the codeword $\beta$
of $\Vir_{1/2}^{\otimes 16n}$ gives irreducible sectors of
$\Vir_{1/2}^{\otimes 16n} \rtimes C$ with statistical dimensions $1$
since the $\alpha^{\pm}$-induction from $\Vir_{1/2}^{\otimes 16n}$ to
$\Vir_{1/2}^{\otimes 16n} \rtimes C_r$ and the double
$\alpha^{\pm}$-induction from $\Vir_{1/2}^{\otimes 16n}\rtimes C$ to
and to $\Vir_{1/2}^{\otimes 16n}\rtimes C_r$
are the same for $r = 2, 3, \dots, p$.

\section{Preliminaries}

\subsection{Local conformal nets}
Denote the collection of all non-dense open intervals in $S^1$ by
$\mathcal{I}$. For each $I$ in $\mathcal{I}$, assign a von Neumann
algebra $\mathcal{A}(I)$ on a Hilbert space $\mathcal{H}$.

\begin{definition}{\rm 
The family $\{\mathcal{A}(I)\}_{I \in \mathcal{I}}$ is called
a \emph{local net}, simply denoted by $\mathcal{A}$ , if it satisfies
the following conditions for all $I, I_1$ and $I_2$ in $\mathcal{I}$.
\begin{enumerate}
\item[(a)] (Isotony) The condition $I_1 \subset I_2$
implies $\mathcal{A}(I_1) \subset \mathcal{A}(I_2)$.
\item[(b)] (Locality) The condition $I_1 \cap I_2 = \emptyset$
implies $[\mathcal{A}(I_1), \mathcal{A}(I_2)] = \{ 0 \}$, where the
Lie bracket denotes the commutator.
\item[(c)] (M\"obius covariance) There exists a strongly continuous
unitary representation $U$ of $\mathrm{PSL}(2,\mathbb{R})$ on $\mathcal{H}$
such that
$$U(g)\mathcal{A}(I)U(g)^{*} = \mathcal{A}(gI), \mathrm{for \; all}\;
g \in \mathrm{PSL}(2,\mathbb{R}), I \in \mathcal{I}.$$
\item[(d)] (Positivity of the energy) Denote the rotational subgroup
of $\mathrm{PSL}(2,\mathbb{R})$ by $R(\cdot)$. If $e^{-i\theta L_0} =
U(R(\theta))$, then $L_0$ is positive.
\item[(e)] (Existence of the vacuum) There exists a $U$-invariant unit
vector $\Omega \in \mathcal{H}$ called a \emph{vacuum vector}.
\item[(f)] (Irreducibility) The von Neumann algebra generated by
all $\mathcal{I}$ is equal to $\mathcal{B}(\mathcal{H})$, i.e.,
$\bigvee_{I \in \mathcal{I}} \mathcal{A}(I) = \mathcal{B}(\mathcal{H})$.
This condition is equivalent to the uniqueness of $\Omega$ up to
phase factor and $\mathcal{A}(I)$ being type $\mathrm{III}_1$ factor
unless it is $\mathbb{C}$.
\end{enumerate}
}\end{definition}

There are some important consequences from this definition.

\begin{proposition}
Let $\mathcal{A}$ be a local net.
\begin{enumerate}
\item (Reeh-Schlieder theorem) $\Omega$ is cyclic and separating
for all $\mathcal{A}(I)$ where $I$ is in $\mathcal{I}$.
\item (Haag's Duality) $\mathcal{A}(I') = \mathcal{A}(I)'$ where $I'$
is the set of interior points of $S^1 \setminus I$ and $\mathcal{A}(I)'$
is the commutant of $\mathcal{A}(I)$ with respect
to $\mathcal{B}(\mathcal{H})$.
\end{enumerate}
\end{proposition}

\begin{definition}{\rm 
The local net $\mathcal{A}$ is called a \emph{local conformal net}
if (c) (M\"obius covariance) is extended to the following condition.
\begin{enumerate}
\item[(c$'$)] (Conformal covariance) $U$ extends to a projective
unitary representation of $\mathrm{Diff}(S^1)$, the group of smooth
orientation preserving diffeomorphisms and
\begin{eqnarray*}
U(g)\mathcal{A}(I)U(g)^{*} &=& \mathcal{A}(gI), \mathrm{for \; all}\;
g \in \mathrm{Diff}(S^1), I \in \mathcal{I}; \\
U(g)xU(g)^{*} &=& x, \mathrm{for \; all}\; g \in \mathrm{Diff}(I'),
x \in \mathcal{A}(I).
\end{eqnarray*}
Here, $\mathrm{Diff}(I)$ is the subset of $\mathrm{Diff}(S^1)$ that
leaves every element in $I'$ invariant.
\end{enumerate}
}\end{definition}

\begin{remark}{\rm
We call the six conditions in the definition of local conformal nets
the axioms of chiral conformal field theories. Local conformal nets
with the locality defined by spacelike separation on $1+1$ Minkowski
space can be restricted to two chiral conformal field theories on the
lines $x=\pm t$ as shown in \cite{Kawahigashi-Longo-Comm}.
}\end{remark}

\begin{definition}{\rm
A \emph{DHR(Doplicher-Haag-Roberts) representation} $\pi$ of a
local conformal net $\mathcal{A}$ on a Hilbert space $\mathcal{K}$
is a map 
$$\mathcal{I} \ni I \mapsto \pi_{I} \subset
\mathcal{B}(\mathcal{K})$$
where $\pi_{I}$ is a normal representation of $\mathcal{A}(I)$.

A DHR representation is \emph{M\"obius covariant} with positive
energy if there exists a projective unitary representation
$U_{\pi}$ of $\mathrm{PSL}(2,\mathbb{R})$ with positive energy
such that for any $I$ in $\mathcal{I}$, $x$ in $\mathcal{A}(I)$
and $g$ in $\mathrm{PSL}(2,\mathbb{R})$
$$U_{\pi}(g) \pi_I (x) U_{\pi} (g)^{*} = \pi_{gI}(U(g)xU(g)^{*}).$$

A DHR representation is \emph{conformal covariant} with positive
energy if there exists a projective unitary representation $U_{\pi}$
of $\mathrm{Diff}^{(\infty)}(S^1)$, the universal cover of
$\mathrm{Diff}(S^1)$, with positive energy such that for any $I$
in $\mathcal{I}$, $x$ in $\mathcal{A}(I)$ and $g$ in
$\mathrm{Diff}^{(\infty)}(S^1)$
$$U_{\pi}(g) \pi_I (x) U_{\pi} (g)^{*} = 
\pi_{\dot{g}I}(U(\dot{g})xU(\dot{g})^{*})$$
where $\dot{g}$ is the image of $g$ in $\mathrm{Diff}(S^1)$.
}\end{definition}

By Doplicher-Haag-Roberts sector theory, $\pi$ corresponds
bijectively to an equivalent class of endomorphisms
$[\rho_I]_{I \in \mathcal{I}}$ on $\mathcal{B}(\mathcal{H})$
where $\rho_I$ is a localized endomorphism of $\mathcal{A}(I)$,
i.e., leaving $\mathcal{A}(I')$ invariant(cf. \cite{Araki}).
The class $[\rho_I]_{I \in \mathcal{I}}$ is often simply denoted
by $\rho$. We call such a class a \emph{DHR sector} or simply
a \emph{sector}. For each $I$ in $\mathcal{I}$, $\rho_I(\mathcal{A}(I))
\subset \mathcal{A}(I)$ is a subfactor by the construction of $\rho$.
We define the \emph{statistical dimension} or the \emph{dimension}
$d_{\rho}$ of $\rho$ as
$$d_{\rho} = ([\mathcal{A}_I:\rho_I(\mathcal{A}(I))])^{\frac{1}{2}}.$$

\subsection{Complete rationality}

Let $\mathcal{A}$ be a local conformal net on a fixed Hilbert
space $\mathcal{H}$. 

\begin{definition}{\rm \cite{Kawahigashi-Longo-Mueger, Longo-Xu}
The local conformal net $\mathcal{A}$ is \emph{completely rational}
if it satisfies the following properties:
\begin{enumerate}
\item (Split property) $\mathcal{A}(I_1) \vee \mathcal{A}(I_2)
\cong \mathcal{A}(I_1) \otimes \mathcal{A}(I_2)$ where $\overline{I_1}$
and $\overline{I_2}$ are disjoint.
\item (Finite $\mu$-index) Split $S^1$ into fours intervals $I_1 ,
I_2 , I_3$ and $I_4$ anti-clockwise. The \emph{$\mu$-index} of
$\mathcal{A}$, $\mu_{\mathcal{A}} = [(\mathcal{A}(I_2) \vee
\mathcal{A}(I_4))' : \mathcal{A}(I_1) \vee \mathcal{A}(I_3)]$, is finite.
\end{enumerate}

The local conformal net $\mathcal{A}$ is \emph{holomorphic} if it is
completely rational and its $\mu$-index is equal to one.
}\end{definition}

\begin{theorem} \cite{Kawahigashi-Longo-Mueger, Longo-Xu}
Let $\mathcal{A}$ be a local conformal net with split property.
If $\mathcal{A}$ has finitely many irreducible DHR sectors with
positive energy up to isomorphism and the statistical dimension of
each irreducible DHR sector is finite, then $\mathcal{A}$ is
completely rational and $\mu_{\mathcal{A}} = \sum^{k}_{i = 1}
d_{\rho_i}$ where $i$ runs over the irreducible DHR sectors
of $\mathcal{A}$.
\end{theorem}

\subsection{The Virasoro net $\mathrm{Vir}_{1/2}$}

The Virasoro algebra is the Lie algebra generated by
$\{L_n\}_{n \in \mathbb{Z}}$ and $\mathbf{c}$ such that
$$[L_n , L_m ] = (m-n)L_{m+n} + \frac{m^3-m}{12} \delta_{m+n,0}
\mathbf{c}$$;
and $[L_n, \mathbf{c}] = 0$ for every $n$ in $\mathbb{Z}$. By
\cite{Friedan-Qiu-Shenker} and \cite{Goddard-Kent-Olive}, either
the central charge $c \geq 1$ or
$$c = 1 - \frac{6}{m(m+1)}, \; m =2,3,4,\dots\;$$
in an irreducible unitary representation.
The lowest eigenvalue $h$ of $L_0$ is called the \emph{conformal weight}
on the unitary representations. By \cite{Goddard-Kent-Olive}, for a
given $c$, the possible values of $h$ are
$$h = \frac{((m+1)p-mq)^2 - 1}{4m(m+1)}, \; \mathrm{where} \; p \in
\{1,2,\dots,m-1\}, q \in \{1,2,\dots,m\}.$$

For $c = 1/2$, there are three unitary representations with conformal
weights $0$, $1/16$ and $1/2$, respectively. Let $U$ be the unitary
representation with conformal weight $0$. Define the Virasoro net
$\mathrm{Vir}_{1/2}$ by (cf. \cite{Kawahigashi-Longo-Annals,Xu-coset}):
$$\mathrm{Vir}_{1/2}(I) = U(\mathrm{Diff}(I))''.$$

By \cite{Kawahigashi-Longo-Annals}, $\mathrm{Vir}_{1/2}$ is
completely rational. There are three inequivalent irreducible DHR
sectors arising from the unitary representations of the Virasoro
algebra with conformal weights $0$, $1/16$, $1/2$, respectively.
They have statistical dimensions $1, \sqrt{2}, 1$, respectively,
and hence the $\mu$-index of $\Vir_{1/2}$ is $4$. Denote these
sectors by $\lambda_0$, $\lambda_{1/16}$ and $\lambda_{1/2}$,
respectively. They obey the following fusion rules where $\lambda(0)$
is the identity sector.
\begin{eqnarray*}
\lambda_{1/2} \circ \lambda_{1/2} &=& \lambda_0,
\; \lambda_{1/2} \circ \lambda_{1/16} = \lambda_{1/16},\\
\lambda_{1/16} \circ \lambda_{1/16} &=& \lambda_{0}
\oplus \lambda_{1/2}.
\end{eqnarray*}

Denote the conformal weight by $h$. By spin-statistics theorem,
we have the conformal spin $\omega = e^{2 \pi i h}$ \cite{Guido-Longo}.
The $S$-matrix of $\mathrm{Vir}_{1/2}$ is given as follows. (See
\cite{DiFrancesco-Mathieu-Senechal}, for example.)
\begin{enumerate}
\item[] $S = $ \(\begin{pmatrix}
\frac{1}{2} & \frac{\sqrt2}{2} & \frac{1}{2} \\
\frac{\sqrt2}{2} & 0 & -\frac{\sqrt2}{2} \\
\frac{1}{2} & -\frac{\sqrt2}{2} & \frac{1}{2}
\end{pmatrix}\)
\end{enumerate}
in the order of $\lambda_0$, $\lambda_{1/16}$ and $\lambda_{1/2}$.

\subsection{Binary codes}

Here, we define some notations on binary codes that we will use
later (cf., \cite{Betsumiya-Munemasa}). A \emph{binary code} $C$ of
\emph{length} $n$  is a subgroup of $\mathbb{Z}^{n}_{2}$. A member
$c$ in $C$ is called a \emph{word} or a \emph{codeword} in $C$. The
\emph{dimension} of $C$, denoted by $\mathrm{dim}(C)$, is an integer
$k$ such that $C$ is isomorphic to $\mathbb{Z}^{k}_{2}$. The
\emph{support} of $c = (x_1, x_2, \dots, x_n)$ is the set $\mathrm{supp}(c)
= \{ i = 1, 2,\dots,n \mid c_i = 1\}$. The \emph{Hamming weight}
or simply the \emph{weight} of $c$ is the cardinality of
$\mathrm{supp}(c)$. The code $C$ is called \emph{even}, \emph{doubly even}
and \emph{triply even} if the Hamming weights of the words in $C$ are
divisible by $2$, $4$ and $8$, respectively. We denote the all-zero word
and all-one word of length $n$ by $(0)_n$ and $(1)_n$, respectively.
Let $c_1 = (x_1, x_2,\dots, x_n)$ and $c_2 = (y_1, y_2, \dots, y_n)$ be
in $C$. Define the \emph{inner product} as
$$c_1 \cdot c_2 = \sum^{n}_{i=1} x_i y_i$$
where the multiplication is the multiplication of real numbers and the
addition is the addition in $\mathbb{Z}_2$. The \emph{dual} $C^{\perp}$
of $C$ is defined as
$$C^{\perp} = \{ c' \in \mathbb{Z}^{n}_{2} \mid c \cdot c' = 0,
\forall c \in C\}$$
and $\mathrm{dim}(C) + \mathrm{dim}(C^{\perp}) = n$. Let $D$ be a binary
code of length $m$ and $d$ = $(z_1 , z_2 , \dots, z_m )$ be a word in $D$.
The \emph{direct sums} between words and codes are defined as follows:
\begin{eqnarray*}
c \oplus d &=& (x_1 , x_2 , \dots , x_n , d_1 , d_2 , \dots , d_m); \\
C \oplus D &=& \{ c \oplus d \mid c \in C \ \mathrm{and} \ d \in D \}.
\end{eqnarray*}
Suppose that $c_1 , c_2 , \dots, c_k $ are words in $\mathbb{Z}^{n}_2$.
The code generated by $c_1 , c_2 , \dots, c_k$ is denoted by
$\langle c_1, c_2 , \dots, c_k \rangle$.

\section{Framed nets}

Here we recall some basics of framed nets in
\cite{Kawahigashi-Longo-Advances}. $\mathrm{Vir}^{\otimes n}_{1/2}$
has $3^n$ inequivalent irreducible sectors. We label the DHR sectors
of $\Vir_{1/2}$ as $0,1/16,1/2$ using the conformal weights as usual. 
Then each $$\lambda=(\lambda_1,\lambda_2,\dots,\lambda_{n})\in
\{0,1/16, 1/2\}^{n}$$ represents an irreducible DHR sector of
${\Vir_{1/2}}^{\otimes n}$ and any irreducible DHR sector of
${\Vir_{1/2}}^{\otimes n}$ is of this form. The fusion rules
follow componentwise. The conformal weight of $\lambda$ is
$\sum^{n}_{i=1} \lambda_i$. The statistical dimension of  $\lambda$
is $2^{\frac{k}{2}}$ where $k$ is the number of entries
$\lambda_i = 1/16$. Hence, the $\mu$-index is $4^n$. We can extend
the net $\mathrm{Vir}^{\otimes n}_{1/2}$ by the following lemma.

\begin{lemma} \cite{Kawahigashi-Longo-Advances}
Let $\mathcal{A}$ be a local conformal net and $\{\lambda_i\}_i$ be
a finite system of irreducible sectors of $\mathcal{A}$ with
statistical dimension $1$ and conformal spin $1$. Then, the crossed
product of $\mathcal{A}$ by the finite abelian group $G$ given by
$\{\lambda_i\}_i$ produces a local extension of the net $\mathcal{A}$.
We call the extended net a \emph{simple current extension} of
$\mathcal{A}$ denoted by $\mathcal{A} \rtimes G$.
\end{lemma}

In the criteria of the preceding lemma, we can form a simple current
extension of $\mathrm{Vir}^{\otimes n}_{1/2}$ by using the group
$\lambda(C)$ consisting of $$\lambda(c) = 
(c_1 /2 , c_2 /2 ,\dots, c_n /2)$$where $c = (c_1, c_2, \dots, c_n)$
belongs to an even binary code $C$ of length $n$. We call $C$ a
\emph{$1/2$-code} as in \cite{Lam-Yamauchi}. On the other hand,
$\lambda$ with some $\lambda_i = 1/16$ is not an automorphism so a
class of such DHR sectors does not give any simple current
extension immediately.

\begin{definition}{\rm \cite{Kawahigashi-Longo-Advances}
A local conformal net $\mathcal{A}$ is called a \emph{framed net}
if it is an irreducible extension of $\mathrm{Vir}^{\otimes n}_{1/2}$
for some positive integer $n$.
}\end{definition}

The next theorem shows the relationship between framed nets and simple
current extensions of $\mathrm{Vir}^{\otimes n}_{1/2}$.

\begin{theorem} \cite{Kawahigashi-Longo-Advances}
Let $\mathcal{A}$ be a framed net extended from
$\mathrm{Vir}^{\otimes n}_{1/2}$. There exist integers $k, l$ and
actions of $\mathbb{Z}^{k}_2$, $\mathbb{Z}^{l}_2$ such that
$\mathcal{A}$ is isomorphic to a simple current extension of a simple
current extension of $\mathrm{Vir}^{\otimes n}_{1/2}$ as follows.
$$\mathcal{A} \cong (\mathrm{Vir}^{\otimes n}_{1/2}
\rtimes \mathbb{Z}^{k}_2 ) \rtimes \mathbb{Z}^{l}_2 .$$
\end{theorem}

We introduce the proof of the next theorem since its method will
be useful later.

\begin{theorem} \cite{Kawahigashi-Longo-Advances} \label{theorem1}
$\mathrm{Vir}^{\otimes 2}_{1/2} \rtimes C$ where $C = \{ (0,0), (1,1)\}$
has four inequivalent irreducible sectors with conformal weights $0$,
$1/8$, $1/2$ and $1/8$, respectively. The fusion rules are given
by $\mathbb{Z}_4$.
\end{theorem}
\begin{proof}
$\mathrm{Vir}^{\otimes 2}_{1/2}$ has nine irreducible sectors.
The $S$-matrix is given by
$$ S = \begin{pmatrix}
\frac{1}{4} & \frac{\sqrt2}{4} & \frac{\sqrt2}{4} & \frac{1}{4}
& \frac{1}{4} & \frac{1}{2} & \frac{\sqrt2}{4} & \frac{\sqrt2}{4}
& \frac{1}{4}\\
\frac{\sqrt2}{4} & 0 & \frac{1}{2} & -\frac{\sqrt2}{4}
& \frac{\sqrt2}{4} & 0 & -\frac{1}{2} & 0 & -\frac{\sqrt2}{4} \\
\frac{\sqrt2}{4} & \frac{1}{2} & 0 & \frac{\sqrt2}{4}
& -\frac{\sqrt2}{4} & 0 & 0 & -\frac{1}{2} & -\frac{\sqrt2}{4} \\
\frac{1}{4} & -\frac{\sqrt2}{4} & \frac{\sqrt2}{4} &
\frac{1}{4} & \frac{1}{4} & -\frac{1}{2} & \frac{\sqrt2}{4}
& -\frac{\sqrt2}{4} & \frac{1}{4}\\
\frac{1}{4} & \frac{\sqrt2}{4} & -\frac{\sqrt2}{4} &
\frac{1}{4} & \frac{1}{4} & -\frac{1}{2} & -\frac{\sqrt2}{4}
& \frac{\sqrt2}{4} & \frac{1}{4}\\
\frac{1}{2} & 0 & 0& -\frac{1}{2} & -\frac{1}{2} & 0 & 0 & 0
& \frac{1}{2}\\
\frac{\sqrt2}{4} & -\frac{1}{2} & 0 & \frac{\sqrt2}{4} &
-\frac{\sqrt2}{4} & 0 & 0 & \frac{1}{2} & -\frac{\sqrt2}{4} \\
\frac{\sqrt2}{4} & 0 & -\frac{1}{2} & -\frac{\sqrt2}{4} &
\frac{\sqrt2}{4} & 0 & \frac{1}{2} & 0 & -\frac{\sqrt2}{4} \\
\frac{1}{4} & -\frac{\sqrt2}{4} & -\frac{\sqrt2}{4} &
\frac{1}{4} & \frac{1}{4} & \frac{1}{2} & -\frac{\sqrt2}{4}
& -\frac{\sqrt2}{4} & \frac{1}{4}
\end{pmatrix}$$
in the order of $\id$, $(0, 1/16)$, $(1/16,0)$, $(0, 1/2)$,
$(1/2,0)$, $(1/16, 1/16)$, \\ $(1/16, 1/2)$, $(1/2, 1/16)$
and $(1/2, 1/2)$.

In this case, $\theta = \id \oplus (1/2,1/2)$. Applying
Proposition 3.23 of part $\mathrm{I}$ and Proposition 5.1
of part $\mathrm{III}$ of \cite{Boeckenhauer-Evans} and
Theorem 5.10 of \cite{Boeckenhauer-Evans-Kawahigashi},
$\alpha^{\pm}$-induction of $\id$, $(0, 1/2)$, $(1/2,0)$,
$(1/2, 1/2)$ and $(1/16, 1/16)$ gives irreducible sectors
of $\mathrm{Vir}^{\otimes 2}_{1/2} \rtimes C$. By
$\langle \alpha^{\pm}_{\rho_1} , \alpha^{\pm}_{\rho_2}
\rangle_{\Vir^{\otimes 2}_{1/2} \rtimes C} = \langle \theta
\circ \rho_1 , \rho_2 \rangle_{\Vir^{\otimes 2}_{1/2}}$,
the $\alpha^{\pm}$-induction of $(0, 1/2)$ is the same as
that of $(1/2,0)$. Likewise, the $\alpha^{\pm}$-induction of
$\id$ is the same as that of $(1/2 , 1/2)$.
The $\alpha^{\pm}$-induction of $(1/16, 1/16)$ has the statistical
dimension $2$. By \cite[Proposition 24]{Kawahigashi-Longo-Mueger},
the $\mu$-index of $\mathrm{Vir}^{\otimes 2}_{1/2} \rtimes C$ is $4$.
So $\mathrm{Vir}^{\otimes 2}_{1/2} \rtimes C$ has four inequivalent
irreducible sectors of statistical dimension $1$. Since the sectors
have statistical dimension $1$, there are two possibilities for
the fusion rules which are $\mathbb{Z}^2_2$ and $\mathbb{Z}_4$.
The latter case occurs since by \cite{Guido-Longo} the conformal
spins are preserved and $\mathbb{Z}^2_2$ violates section $3$ of
\cite{Rehren}.
\end{proof}

The preceding theorem shows that the $\alpha^{\pm}$-induction of
$\lambda$ with some $\lambda_i = 1/16$ may give irreducible
sectors of statistical dimension $1$. It may be possible to find
a binary code $D$ that acts on these irreducible sectors such that
the simple current extension $\mathrm{Vir}^{\otimes n}_{1/2}
\rtimes C \rtimes D$ occurs. We call $D$ a \emph{$1/16$-code},
its element $d \in D$ a \emph{$1/16$-word} or a \emph{$\tau$-word}
and the pair $(C,D)$ \emph{structure codes} as in \cite{Lam-Yamauchi}.
We have the following theorem that resembles framed vertex
operator algebras.

\begin{proposition}
Suppose that the simple current extension
$\mathrm{Vir}^{\otimes n}_{1/2} \rtimes C \rtimes D$
is well-defined. The following statements hold.
\begin{enumerate}
\item $C$ is  an even binary code and $D$ is a triply even
binary code, i.e., the Hamming weights of the elements in $D$
are divisible by 8.
\item $C \subset D^{\perp}$.
\end{enumerate}
\end{proposition}
\begin{proof}
\begin{enumerate}
\item This follows from the definition of simple current extension
since only DHR sectors with triply even $1/16$-words can have the
conformal spins 1.
\item Suppose that $C$ is not a subset of $D^{\perp}$ . Then
there exist $c$ in $C$ and $d$ in $D$ such that the cardinality
of $\mathrm{supp}(c) \cap \mathrm{supp}(d)$ is odd. The $S$-matrix
element associating to $c$ and $d$ is negative as the $S$-matrix
of $\mathrm{Vir}^{\otimes n}_{1/2}$ is the $n$th tensor power of
the $S$-matrix of $\mathrm{Vir}_{1/2}$. So the $\alpha^{\pm}$-induction
of $\lambda$ such that the $1/16$-word is $d$ does not give
irreducible sectors of $\mathrm{Vir}^{\otimes n}_{1/2} \rtimes C$.
\end{enumerate}
\end{proof}

\section{Holomorphic framed nets}

There some more restrictions on the codes $C$ and $D$ if the framed
net $\mathrm{Vir}^{\otimes n}_{1/2} \rtimes C \rtimes D$ is holomorphic.
Let $C$ be a binary code of length $16n$ satisfying the
condition in \cite{Lam-Yamauchi} where $D=C^\perp$ as follows.
\begin{enumerate}
\item The length of $D$ is $16n$ where $n$ is a positive integer.
\item The code $D$ is triply even.
\item The word $(1)_{16n}$ is in $D$.
\end{enumerate}
Our aim is to construct a holomorphic local conformal net with
the central charge $c=8n$ and
structure codes $(C,D)$ in the form of
${\Vir_{1/2}}^{\otimes 16n}\rtimes C\rtimes D$.
Let $C$ be isomorphic to ${\mathbb Z}_2^k$ as an abstract group.

Recall that the $S$-matrix of ${\Vir_{1/2}}$ is given as
$$\left(\begin{array}{ccc}
\frac{1}{2} & \frac{\sqrt2}{2} & \frac{1}{2} \\
\frac{\sqrt2}{2} & 0 & -\frac{\sqrt2}{2} \\
\frac{1}{2} & -\frac{\sqrt2}{2} & \frac{1}{2}
\end{array}\right),$$
where the orders of rows and columns are given by conformal
weights as $0,1/16,1/2$, respectively.

Note that since $D$ contains $(1)_{16n}$, all the codewords
in $C$ are even, so 
$\lambda(C)=\{\lambda(c)\mid c\in C\}$ naturally gives an
action of $C$ on ${\Vir_{1/2}}^{\otimes 16n}$ since the conformal
spin of the DHR sector $\lambda_{1/2}$ is $-1$.  We then have
a crossed product net
${\Vir_{1/2}}^{\otimes 16n}\rtimes C$ through this action
by \cite[Lemma 2.1]{Kawahigashi-Longo-Advances}.
This is an operator algebraic counterpart of the code 
vertex operator algebra.

The next step is to find an appropriate action of $D$ on 
${\Vir_{1/2}}^{\otimes 16n}\rtimes C$ and to prove that
${\Vir_{1/2}}^{\otimes 16n}\rtimes C\rtimes D$ has a 
structure code $(C,D)$.  (It is trivial that this net is
holomorphic since its $\mu$-index is equal to
$4^{16n}/{|C||D|}^2=1$.)

For this purpose, we need to know the representation theory
of \\　${\Vir_{1/2}}^{\otimes 16n}\rtimes C$.
Note that the dual canonical endomorphism $\theta$ for the
inclusion 
$${\Vir_{1/2}}^{\otimes 16n}\subset
{\Vir_{1/2}}^{\otimes 16n}\rtimes C$$ is given as
$\bigoplus_{c\in C}\lambda(c)$, so for a general $\lambda\in
\{0, 1/16, 1/2\}^{16n}$, we have
$\alpha^+(\lambda)=\alpha^-(\lambda)$ if and only if
$\varepsilon^+(\lambda,\theta)=\varepsilon^-(\lambda,\theta)$
if and only if $\tau(\lambda)\in C^\perp=D$, where $\tau(\lambda)$
is the $\tau$-word.
(This follows from the form of the $S$-matrix above as the $S$-matrix
of $\mathrm{Vir}^{\otimes 16n}_{1/2}$ is the $16n$th tensor power
of the $S$-matrix of $\mathrm{Vir}_{1/2}$.)

From now on, we consider only $\lambda$ in $D$.  Then by
\cite[Theorem 5.10]{Boeckenhauer-Evans-Kawahigashi}, 
$\alpha^\pm(\lambda)$ gives a (possibly reducible) DHR sector
of ${\Vir_{1/2}}^{\otimes 16n}\rtimes C$.  
Since we always have $\alpha^+_\lambda=\alpha^-_\lambda$,
we simply write $\alpha_\lambda$.  We are going to prove
that all the irreducible DHR sectors of this local conformal
net arises from irreducible decomposition of 
$\alpha_\lambda$ of $\lambda\in D$.

First note $\langle \alpha_\lambda,\alpha_\mu\rangle
=\langle \lambda\mu,\theta\rangle$, and this implies that
$\langle \alpha_\lambda,\alpha_\mu\rangle\neq 0$ if and
only if $\lambda\mu\in \lambda(C)$.

Fix $\beta\in D$ with weight $8j$
and consider a DHR sector $\lambda$ of ${\Vir_{1/2}}^{\otimes 16n}$ with
$\tau(\lambda)=\beta$.  The number of such $\lambda$'s is
$2^{16n-8j}$.  As in \cite{Lam-Yamauchi}, set
$$C_\beta=\{c\in C\mid \supp (c) \subset \beta\}.$$
Note that we have $|C_\beta|\ge 2^{4j}$ by 
\cite[Remark 6]{Lam-Yamauchi}.  The number of distinct
$\alpha_\lambda$ is now $2^{16n-8j}|C_\beta|/|C|$.  We
also have $d_{\lambda}=d_{\alpha_\lambda}=2^{4j}$.
Note that $$\langle \alpha_\lambda, \alpha_\lambda\rangle = 
\sum_{c\in C}\langle \lambda^2,c\rangle=|C_\beta|\ge 2^{4j}.$$
Suppose that we have an irreducible decomposition
$\alpha_\lambda=\bigoplus_i m_i \sigma_i$, where $m_i$ is 
the multiplicity of the irreducible DHR sector $\sigma_i$.
We then have $\langle \alpha_\lambda, \alpha_\lambda\rangle=
\sum_i m_i^2$.  Set $d_{\sigma_i}=d_i$ and
consider the possible lowest value of $\sum_i
d_i^2$.  Since we have $\sum_i m_i d_i=2^{4j}$, the
possible lowest value of $\sum_i d_i^2$ is $2^{8j}/|C_\beta|$ by
the Cauchy-Schwarz inequality
and this happens when we have $|C_\beta|/2^{4j}=m_i/d_i$ for
all $i$.  If we have this lowest value for all $\lambda$, the 
total contribution to the $\mu$-index of 
${\Vir_{1/2}}^{\otimes 16n}\rtimes C$ is equal to
$2^{16k}/|C|=|D|$.  Now the number of $\beta$ is $|D|$, so
all the contribution to the $\mu$-index of 
${\Vir_{1/2}}^{\otimes 16n}\rtimes C$ arising in this way
is at least $|D|^2$, which is the right $\mu$-index of
${\Vir_{1/2}}^{\otimes 16n}\rtimes C$.  This shows that
we already have all the irreducible DHR sectors of 
${\Vir_{1/2}}^{\otimes 16n}\rtimes C$, and we
have $|C_\beta|/2^{4j}=m_i/d_i$ for all $\beta$, $\lambda$ 
and $i$.  Note that this equality implies $d_i$ is rational, but 
it is also an algebraic integer, so each $d_i$ has to be an
integer.

\begin{lemma}\label{lemma2}
Fix $\lambda$ and $\beta$ as above, and consider the
irreducible decomposition
$\alpha_\lambda=\bigoplus_i m_i \sigma_i$.
Then all $m_i$'s are equal.
\end{lemma}

\begin{proof}
By the fusion rules and $\langle \alpha_\lambda,\alpha_\mu\rangle
=\langle \lambda\mu,\theta\rangle$, we have
$$\alpha_\lambda \circ \alpha_{\bar\lambda}=\alpha_{\lambda^2}
=\bigoplus_m |C_\beta| \alpha_{\mu_m},$$
where the number of $m$'s is $2^{8j}/|C_\beta|$, $|C_\beta|$
is the multiplicity, all $\mu_m$ are
mutually inequivalent, and each $\mu_m$ is in
$\{0,1/2\}^{16n}$, hence has a dimension $1$.  Now take
$\sigma_i, \sigma_l$ appearing in the irreducible
decomposition of $\alpha_\lambda$.
Then $\sigma_i\cdot\bar\sigma_l$ is a direct sum of
DHR sectors with dimension $1$ of the form $\alpha_\mu$,
$\mu\in \{0,1/2\}^{16n}$ possibly with some multiplicity.
Choose
one such $\mu$.  Then the Frobenius reciprocity implies
$\langle \sigma_i,\sigma_l \alpha_\mu \rangle\ge 1$, but
the dimension of $\alpha_\mu$ is $1$, so we have
$\sigma_i=\sigma_l \alpha_\mu$ (cf. \cite{Boeckenhauer-Evans-Kawahigashi}).
Since $\supp(\mu)$ is
contained in $\beta=\{h\mid \lambda_h=1/16\}$, we have
$\alpha_{\lambda\mu}=\alpha_{\lambda}$, which implies 
the equality of the multiplicities, $m_i=m_l$.  That is,
for fixed $\lambda$, all the $m_i$'s in the
irreducible decomposition
$\alpha_\lambda=\bigoplus_i m_i \sigma_i$ are equal,
and we simply write 
$\alpha_\lambda=m \bigoplus_i \sigma_i$.  This also gives
$d_{\sigma_i}=m2^{4j}/|C_\beta|$.
\end{proof}

Here we prove the following general lemma on a modular
tensor category.

\begin{lemma}\label{lemma1}
Fix a modular tensor category and suppose that the
dimensions of the irreducible objects are all $1$ and
the conformal spins of the irreducible objects are all $\pm1$.
Then all the nontrivial elements in this modular tensor
category has order $2$.
\end{lemma}

\begin{proof}
For irreducible objects
$\lambda,\mu,\nu$ in this tensor category, we have
$$Y_{\lambda\mu}=
\frac{\omega_\lambda \omega_\mu}{\omega_{\lambda\mu}}=\pm1$$
since $N_{\lambda\mu}^\nu\neq0$ only if $\lambda\mu=\nu$.  Then
$S_{\lambda\mu}=w^{-1/2}Y_{\lambda\mu}$, and the
Verlinde formula gives
$$N_{\mu\mu}^0=\sum_\lambda
\frac{S_{\mu\lambda}S_{\mu\lambda}S_{0\lambda}^*}
{S_{0\lambda}}=1,$$ so this means
$\mu\mu=0$ for all non-trivial irreducible object $\mu$.
\end{proof}

We want to show that an automorphism group $\Delta$ isomorphic to $D$
on $\mathrm{Vir}^{\otimes 16n}_{1/2} \rtimes C$ exists which implies
the existence of the holomorphic framed net
$\mathrm{Vir}^{\otimes 16n}_{1/2} \rtimes C \rtimes D$.

\begin{proposition}\label{prop1}
In the above setting, there exists a set $\Delta$ of
irreducible DHR sectors of
${\Vir_{1/2}}^{\otimes 16n}\rtimes C$
satisfying the following.
\begin{enumerate}
\item Each element in $\Delta$ has dimension $1$.
\item Each element in $\Delta$ has conformal spin $1$.
\item The set $\Delta$ is closed under the sector
multiplication and conjugation.
\item For each $\beta\in D$, we have exactly one
DHR sector in $\Delta$ which arises from the
irreducible decomposition of any $\alpha_\lambda$ 
where the $\tau$-word of $\lambda$ is $\beta$, and 
through this, we have a group isomorphism of $\Delta$
and $D$.
\end{enumerate}
\end{proposition}

\begin{proof} We divide the problem into three cases according to
the dimension of $D$.
\subsubsection*{Case 1: $\mathrm{dim}(D) = 1$}

The code $D$ has only two words which are $(0)_{16n}$ and $(1)_{16n}$
so $D = \langle(1)_{16n}\rangle$. By looking at the $S$-matrix and
following the method given in the proof of theorem \ref{theorem1},
the $\alpha$-induction of the following sectors on 
$\mathrm{Vir}^{\otimes 16n}_{1/2}$ give irreducible DHR sectors on
$\mathrm{Vir}^{\otimes 16n}_{1/2} \rtimes
\langle(1)_{16n}\rangle^{\perp}$: $\id$, $(0)_{16n-1} \oplus (1/2)$
and $(1/16)_{16n}$ with statistical dimensions $1$, $1$ and $2^{8n}$,
respectively. $\langle(1)_{16n}\rangle^{\perp}$ has $2^{16n-1}$
words so the $\mu$-index of $\mathcal{A} \cong
\mathrm{Vir}^{\otimes 16n}_{1/2} \rtimes
\langle(1)_{16n}\rangle^{\perp}$ is $4$.

There are the following possibilities.
\begin{enumerate}
\item The $\alpha$-induction of $(1/16)_{16n}$ splits into irreducible DHR
sectors of statistical dimension $2^k$, for some positive integer $k$.
\item The $\alpha$-induction of $(1/16)_{16n}$ splits into a number
different from two of inequivalent irreducible DHR sectors of
statistical dimension $1$.
\item The $\alpha$-induction of $(1/16)_{16n}$ splits into two
inequivalent irreducible DHR sectors of statistical dimension $1$
with the multiplicity $2^{8n-1}$.
\end{enumerate}
Case $3.$ occurs by the value of the $\mu$-index. Denote the irreducible
DHR subsectors of the $\alpha$-induction of $(1/16)_{16n}$ by
$\sigma_1$ and $\sigma_2$. The conformal spins of $\sigma_1$ and
$\sigma_2$ are $1$. By lemma \ref{lemma1}, all the four irreducible
DHR sectors have order 2. We have $\Delta = \{ \id, \sigma_i\}$
where $i$ can be either $1$ or $2$.

\subsubsection*{Case 2: $\mathrm{dim}(D) = 2$}

There are four words in $D$. Suppose that $d$ is a word in $D$
different from $(0)_{16n}$ and $(1)_{16n}$. Then, the other word in
$D$ is $(1)_{16n} + d$. The Hamming weight of $d$ is $8i$,
$i = 1, 2, \dots, 2n-1$, if and only if the Hamming weight of
$(1)_{16n} + d$ is $16n - 8i$. Rewrite $D$ as
\begin{eqnarray*}
D &=& \langle (1)_{8i} \oplus (0)_{16n-8i}, (0)_{8i} \oplus
(1)_{16n-8i} \rangle \\
&=& \langle (1)_{8i} \rangle \oplus \langle (1)_{16n-8i} \rangle.
\end{eqnarray*} 
Then, $D^{\perp} = \langle (1)_{8i} \rangle^{\perp} \oplus \langle
(1)_{16n-8i} \rangle^{\perp}$. By the definition of crossed product
von Neumann algebras,
\[
\mathrm{Vir}^{\otimes 16n}_{1/2} \rtimes (\langle (1)_{8i}
\rangle^{\perp} \oplus \langle (1)_{16n-8i} \rangle^{\perp})
\]
\[
\cong (\mathrm{Vir}^{\otimes 8i}_{1/2} \rtimes \langle (1)_{8i}
\rangle^{\perp}) \otimes (\mathrm{Vir}^{\otimes (16n-8i)}_{1/2} 
\rtimes \langle (1)_{16n-8i} \rangle^{\perp}).
\]
Since $i$ is arbitrary, it is sufficient to consider only
$\mathrm{Vir}^{\otimes 8i}_{1/2} \rtimes \langle (1)_{8i}
\rangle ^{\perp}$. Using the same argument as the case
$\mathrm{dim}(D) = 1$, $\mathrm{Vir}^{\otimes 8i}_{1/2}
\rtimes \langle (1)_{8i} \rangle^{\perp}$ totally has four
irreducible sectors of quantum dimension $1$ where $\id$ and
$(0)_{8i-1} \oplus (1/2)_1 $ induce one sector each and the
$\alpha$-induction of $(1/16)_{8i}$ splits into two inequivalent
irreducible sectors of statistical dimension $1$. Denote the
splitted irreducible subsectors by $\beta_1$ and $\beta_2$.
By lemma \ref{lemma1}, the four irreducible sectors obey the
fusion rules given by $\mathbb{Z}^{2}_2$. $\beta_1$ and
$\beta_2$ satisfy the following equation:
$$\beta_1 \alpha_{(0)_{8i-1} \oplus (1/2)_1} = \beta_2 .$$
For $\mathrm{Vir}^{\otimes (16n-8i)}_{1/2} \rtimes
\langle (1)_{16n-8i} \rangle^{\perp}$, denote the splitted
irreducible subsectors of the $\alpha$-induction of $(1/16)_{16n-8i}$
by $\beta'_1$ and $\beta'_2$. $\beta'_1$ and $\beta'_2$ satisfy
the following equation:
$$\beta'_1 \alpha_{(0)_{16n-8i-1} \oplus (1/2)_1} = \beta'_2 .$$

By the representation theory of $\mathrm{Vir}^{\otimes 8i}_{1/2}
\rtimes \langle (1)_{8i} \rangle^{\perp}$ and \\ 
$\mathrm{Vir}^{\otimes (16n-8i)}_{1/2} \rtimes \langle (1)_{16n-8i}
\rangle^{\perp}$ above, we have $4$ irreducible DHR sectors with
dimension $1$ and the $\tau$-word $(0)_{16n}$ with conformal
spins $1, -1, -1, 1$. We also have $4$ irreducible DHR sectors
with dimension $1$ with the $\tau$-word 
$(1)_{8i} \oplus (0)_{16n-8i}$, and their conformal spins
are $1,1,-1,-1$ by the same argument. We also get the same
conclusion for the $\tau$-word
$(0)_{8i} \oplus (1)_{16n-8i}$.
For $\lambda$'s with their $\tau$-words $(1)_{16n}$,
we have only one $\alpha_\lambda$ and its contribution to
the $\mu$-index is $4$.
We note that all have conformal spin $1$.
If we multiply an irreducible 
DHR sector arising from $\alpha_\lambda$ with
$\lambda$'s $\tau$-word $(1)_{8i} \oplus (0)_{16n-8i}$
and another such a sector arising from 
$\alpha_\lambda$ with
$\lambda$'s $\tau$-word $(0)_{8i} \oplus (1)_{16n-8i}$,
we obtain an irreducible DHR sector of dimension 1
arising from $\alpha_\lambda$ with
$\lambda$'s $\tau$-word $(1)_{16n}$.  There are $4$ irreducible
DHR sectors with the $\tau$-word $(1)_{16n}$. We thus see that
all the statistical dimensions of
all the irreducible DHR sectors are $1$, and we can apply
Lemma \ref{lemma1}.  Then we first choose the identity
sector $\id$.  We next choose $\sigma_1$ with dimension 1 and
conformal spin 1 from an irreducible decomposition of 
$\alpha_\lambda$ with
$\lambda$'s $\tau$-word $(1)_{8i} \oplus (0)_{16n-8i}$. This is
possible for both even and odd $i$.
We also choose $\sigma_2$ with dimension 1 and
conformal spin 1 from an irreducible decomposition of 
$\alpha_\lambda$ with
$\lambda$'s $\tau$-word $(0)_{8i} \oplus (1)_{16n-8i}$.
Then we set $\Delta=\{\id,\sigma_1,\sigma_2,\sigma_1\sigma_2\}$.

\subsubsection*{Case 3: $\mathrm{dim}(D) \geq 3$}

We now proceed by induction. Suppose we have $D$ and we
have proved the proposition for the case where the order is
$D$ is smaller.

Choose $\beta\in D$ with $\beta\neq(0)_{16n}$ and
$\beta\neq(1)_{16n}$.
Choose an irreducible DHR sector $\sigma$ appearing in
the decomposition of some $\alpha_\lambda$ with $\lambda$'s
$\tau$-word $\beta$.  Suppose the dimension $d_\sigma$ is
larger than $1$, and we will derive a contradiction.
As in the proof of Lemma \ref{lemma2},
we know that $\sigma\bar\sigma$ decomposes into a direct
sum of irreducible DHR sectors of dimension $1$.  The
Frobenius reciprocity shows that all of these have
multiplicity $1$.  This means that
the endomorphism $\sigma$ gives a crossed product
subfactor with an abelian group of order $d_\sigma^2$.
Let $G$ be this abelian group.  This is a subgroup
of ${\mathbb Z}_2^{16n}/C$. 

First assume all the irreducible DHR sectors in $G$
have conformal spin $1$.
We then construct a decreasing sequence 
$$
D=D_1\supset D_2\supset\cdots D_{p-1}=\langle \beta,
(1)_{16n} \rangle
 \supset D_p = \langle
(1)_{16n} \rangle
$$
and an increasing sequence 
$$C=C_1\subset C_2 \subset \cdots C_{p-1}\subset C_p,$$
where $[D_r:D_{r-1}]=2$ and $D_r^\perp=C_r$ and
some $C_r$ has the property $C_r/C=G$.  This is easy
by choosing $D_p, D_{p-1}, D_{p-2},\dots$ in this order.
Note that all $D_r$ are triply even codes containing 
$(1)_{16n}$ so that we can apply the above general
argument to $D_r$.
Consider the $\alpha$-induction of $\sigma$ from 
${\Vir_{1/2}}^{\otimes 16n}\rtimes C$ to
${\Vir_{1/2}}^{\otimes 16n}\rtimes C_r$.  By the
$\alpha \sigma$-reciprocity in
\cite{Boeckenhauer-Evans} and that fact that $G$ is an
abelian group, we know that this $\alpha$-induction
produces an irreducible DHR sector of dimension
equal to $d_\sigma$, which is bigger than $1$.  This
$\alpha$-induction is equal to the $\alpha$-induction of
$\lambda$ from ${\Vir_{1/2}}^{\otimes 16n}$ to
${\Vir_{1/2}}^{\otimes 16n}\rtimes C_r$ by
\cite{Longo-Rehren}, so this contradicts the induction
hypothesis.

Next consider the case some of the
irreducible DHR sectors in $G$
have conformal spin $-1$.
Note that the conformal spins
are multiplicative on $G$ because they arise from
$\alpha$-inductions of DHR sectors in $\{0,1/2\}^{16n}$.
Then the extension
${\Vir_{1/2}}^{\otimes 16n}\rtimes C\rtimes G$ is rewritten
as ${\Vir_{1/2}}^{\otimes 16n}\rtimes C\rtimes \tilde G
\rtimes {\mathbb Z}_2$, where $\tilde G$ is a maximal subgroup
giving a local extension
of ${\mathbb Z}_2^{16n}/C$ and the last crossed product
by ${\mathbb Z}_2$ is a Fermionic extension (the locality is replaced by
graded locality, cf. \cite{Carpi-Kawahigashi-Longo}).
In this case the order of $G$ is $d^{2}_\sigma \ge 4$, so
the group $\tilde G$ is nontrivial.
We again construct a decreasing sequence 
$$D=D_1\supset D_2\supset\cdots D_{p-1}\supset D_p$$
and an increasing sequence 
$$C=C_1\subset C_2 \subset \cdots C_{p-1}\subset C_p$$
in a similar way to the above case,
where now some $C_r$ has the property $C_r/C=\tilde G$.
Now again
consider the $\alpha$-induction of $\sigma$ from 
${\Vir_{1/2}}^{\otimes 16n}\rtimes C$ to
${\Vir_{1/2}}^{\otimes 16n}\rtimes C_r\rtimes {\mathbb Z}_2$.
We again know that this $\alpha$-induction
produces an irreducible DHR sector of dimension
equal to $d_\sigma$, which is bigger than $1$.  This
$\alpha$-induction is equal to the $\alpha$-induction of
$\lambda$ from ${\Vir_{1/2}}^{\otimes 16n}$ to
${\Vir_{1/2}}^{\otimes 16n}\rtimes C_r\rtimes {\mathbb Z}_2$.
Every irreducible DHR sector of 
${\Vir_{1/2}}^{\otimes 16n}\rtimes C_r$ has dimension $1$ by
the induction hypothesis, so its $\alpha$-induction to 
${\Vir_{1/2}}^{\otimes 16n}\rtimes C_r\rtimes {\mathbb Z}_2$
again has index $1$ and this is a contradiction.

We have thus proved that all the irreducible DHR sectors
of ${\Vir_{1/2}}^{\otimes 16n}\rtimes C$ have dimension $1$.
Then we choose the generators $\{\beta_s\}$ of $D$, and for
each $\beta_s$, we choose an irreducible DHR sector
$\sigma_s$ with conformal spin $1$ appearing in the
irreducible decomposition of $\alpha_\lambda$ with $\lambda$'s
$\tau$-word $\beta_s$.  Then the group generated
by the irreducible DHR sectors $\beta_s$ gives the desired
set $\Delta$.
\end{proof}

Using the preceding proposition, we immediately obtain the following
main theorem.

\begin{theorem}\label{main}
Suppose the codes $C,D$ are as in \cite{Lam-Yamauchi}.
Then we have a holomorphic
framed local conformal net with structure
codes $(C,D)$.
\end{theorem}

\begin{proof}
We simply make a two-step crossed product
${\Vir_{1/2}}^{\otimes 16n}\rtimes C\rtimes D$,
where the action of $D$ is given by $\Delta$ in 
Proposition \ref{prop1}.  Construct the simple current extension
of ${\Vir_{1/2}}^{\otimes 16n}\rtimes C$ using $\Delta$. It is
easy to see that the
$\mu$-index is $1$ and the structure codes are $(C,D)$.
\end{proof}

\begin{remark}{\rm
As mentioned earlier, Lam and Yamauchi proved the existence of
holomorphic
framed vertex operator algebras associated to binary codes in
\cite{Lam-Yamauchi}. The classification of the
maximal triply even binary
codes given by Betsumiya and Muenasa in \cite{Betsumiya-Munemasa}
leads to the classification of holomorphic framed vertex operator
algebras extended from $L(1/2 , 0)^{\otimes 16n}$, $n=1,2,3$, by Lam
and Shimakura in \cite{Lam-Shimakura}. Concrete examples of such
vertex operator algebras are constructed and identified.
We have the corresponding local conformal nets, and the vacuum
characters are clearly the same in the both approaches.  At
$c=24$, the vacuum character is uniquely determined and equal
to the $j$-function except for the constant term, so if the
constant terms are different in the setting of \cite{Lam-Shimakura},
the corresponding local conformal nets are also different, but
Lam and Shimakura have examples where the constant terms of
the vacuum characters, which are the dimensions of the weight
1 spaces, are the same, but have different Lie algebra structures.
They are different as vertex operator algebras, and we expect 
that the corresponding local conformal nets are also different,
but we do not have a proof so far.
}\end{remark}

\begin{remark}{\rm
Since our argument relies on only structures of the tensor
categories, it is possible to use $SU(2)_2$ instead of
$\Vir_{1/2}$.  (Note that the $S$-matrices of these two are
the same, and one of the three conformal weights differ by
$1/8$, but we use only triply even codes, so this difference
does not matter.)  However, as shown in 
\cite[Theorem A.2]{Lam}, the vertex operator algebra
corresponding to $SU(2)_2$ itself is an extension of
$L(1/2,0)\otimes L(1/2,0)\otimes L(1/2,0)$, so this 
does not give new examples.
}\end{remark}

\noindent {\bf Acknowledgements.}
Part of this work of Y. K. has been 
done at Universit\`a di Roma ``Tor Vergata''.  He gratefully
acknowledges the hospitality there.  He also thanks R. Longo
and H. Yamauchi for helpful comments.
The second-named author thanks the first-named
author for the supervision.

\end{document}